\documentclass[amsmath,amssymb,aps,prx,notitlepage,superscriptaddress,longbibliography,reprint]{revtex4-2}
\usepackage{graphicx}
\usepackage{dcolumn}
\usepackage{bm}
\usepackage{braket}
\usepackage{amsthm}
\usepackage{epstopdf}
\usepackage{amssymb}
\usepackage{amsmath}
\usepackage{bbold}
\usepackage{hyperref}
\usepackage{comment}
\usepackage{xcolor}
\usepackage[ruled,vlined]{algorithm2e}

\theoremstyle{definition}
\newtheorem{definition}{Definition}
\newtheorem{theorem}{Theorem}

\newtheorem{lemma}{Lemma}

\begin{document}
\preprint{APS/123-QED}
\title{Classical simulation of bosonic linear-optical random circuits beyond linear light cone}
\author{Changhun Oh}%
\email{changhun@uchicago.edu}
\affiliation{Pritzker School of Molecular Engineering, University of Chicago, Chicago, Illinois 60637, USA}
\author{Youngrong Lim}%
\email{sshaep@kias.re.kr}
\affiliation{School of Computational Sciences, Korea Institute for Advanced Study, Seoul 02455, Korea}
\author{Bill Fefferman}
\email{wjf@uchicago.edu}
\affiliation{Department of Computer Science, University of Chicago, Chicago, Illinois 60637, USA}
\author{Liang Jiang}
\email{liang.jiang@uchicago.edu}
\affiliation{Pritzker School of Molecular Engineering, University of Chicago, Chicago, Illinois 60637, USA}
\date{\today}
\begin{abstract}
Sampling from probability distributions of quantum circuits is a fundamentally and practically important task which can be used to demonstrate quantum supremacy using noisy intermediate-scale quantum devices.
In the present work, we examine classical simulability of sampling from the output photon-number distribution of linear-optical circuits composed of random beam splitters with equally distributed squeezed vacuum states and single-photon states input. 
We provide efficient classical algorithms to simulate linear-optical random circuits and show that the algorithms' error is exponentially small up to a depth less than quadratic in the distance between sources using a classical random walk behavior of random linear-optical circuits.
Notably, the average-case depth allowing an efficient classical simulation is larger than the worst-case depth limit, which is linear in the distance.
Besides, our results together with the hardness of boson sampling give a lower-bound on the depth for constituting global Haar-random unitary circuits.
\end{abstract}

\maketitle

\section{Introduction}
Quantum computers are believed to provide a computational power that classical computers cannot achieve.
The ultimate goal in the field of quantum computation is implementing a fault-tolerant and universal quantum computer to solve classically intractable and practically important problems, such as integer factorization \cite{shor1994algorithms} and simulation of the real-time dynamics of large quantum systems \cite{lloyd1996universal}.
Since a current technology cannot build a scalable fault-tolerant quantum computer, 
much attention has been paid to demonstrate quantum supremacy using noisy intermediate-scale quantum (NISQ) devices \cite{preskillnisq}.
While numerous interesting tasks have been theoretically proposed and proven to be hard to simulate classically, there has been a rapid development in experiments to manipulate an intermediate size of quantum systems \cite{arute2019quantum, wang2017high, wang2018toward, wang2019boson, zhong2020quantum}.

Sampling from the probability distributions of randomly chosen quantum circuits is one of the promising problems to demonstrate quantum supremacy using NISQ devices.
Especially, various sampling problems have been proven to be hard using classical means unless the polynomial hierarchy collapses; boson sampling \cite{aaronson2011computational, hamilton2017gaussian}, IQP sampling \cite{bremner2011classical}, Fourier sampling \cite{fefferman2015power}, and random circuit sampling \cite{boixo2018characterizing} are the representative examples.
Recently, random circuit sampling was experimentally implemented to demonstrate quantum supremacy using the state-of-the-art superconducting qubits for the task that 
they claim a classical computer would cost a large amount of computational time \cite{arute2019quantum}.
On the other hand, bosonic circuits have also been widely studied theoretically and experimentally since Aaronson's seminal work \cite{aaronson2011computational}, the so-called boson sampling.
Boson sampling is a sampling problem based on linear optics with single-photon input states and photo-detectors.
Because of its relatively simple structure, there have been numerous proof-of-principle experiments of boson sampling \cite{wang2017high, wang2018toward, wang2019boson, zhong2020quantum}.
As the experimental system size of bosonic circuits grows rapidly, establishing adequate conditions for discriminating easiness and hardness of bosonic linear-optical circuits becomes more crucial for the practical demonstration of quantum supremacy. 

Circuit depth is one of the essential parameters that determine a quantum circuit's classical simulability.
In particular, a transition of sampling complexity from easiness to hardness appears as a circuit depth increases. 
On the one hand, when a quantum circuit is shallow, the output state is not entangled enough and easy to classically simulate \cite{vidal2003efficient,napp2019efficient,qi2020efficient}.
On the other hand, when a quantum circuit is deep enough to implement a global Haar-random unitary circuit that generates a large amount of entanglement, sampling from the probability distributions of the output state becomes classically intractable under plausible assumptions \cite{aaronson2011computational, constantdepthpra, hamilton2017gaussian, bouland2019complexity}.
Indeed, under bosonic Hamiltonian dynamics, phase transition behavior of sampling arising from the evolution time has been diagnosed~\cite{deshpande2018dynamical, maskara2019complexity}.
Whereas time-evolution dynamics under a bosonic Hamitonian using the Lieb-Robinson bound \cite{lieb1972finite} has been studied in Refs. \cite{deshpande2018dynamical, maskara2019complexity}, we focus on linear-optical circuits composed of Haar-random beam splitter arrays and employ a mapping from random linear-optical circuits to classical random walk \cite{zhang2020entanglement}.
Investigating random circuits enables us to analyze a generic property of linear-optical circuits.
Remarkably, we show that for random linear-optical circuits, one can efficiently sample their photon-number outcomes for a larger depth  than the one found in Ref. \cite{deshpande2018dynamical}.
Furthermore, currently used linear-optical circuits for boson sampling experiments are not sufficiently deep in implementing a global Haar-random circuit \cite{wang2019boson, zhong2020quantum}.
Besides, reducing the circuit depth when implementing a classically intractable circuit is important to minimize the adversarial effects of photon loss in experiment \cite{aaronson2016bosonsampling, renema2018classical, oszmaniec2018classical, garcia2019simulating, qi2020regimes,oh2021classical}.
Hence, understanding how large circuit depth is required to attain a classically intractable quantum circuit (e.g. a global Haar-random circuit \cite{aaronson2011computational}) is crucial from both practical and theoretical perspectives.

In this paper, we consider bosonic circuits consisting of Haar-random beam splitters between adjacent modes with squeezed vacuum states and single-photon states input.
The setup's importance stems from the computational-complexity hardness result of (Gaussian) boson sampling \cite{aaronson2011computational, hamilton2017gaussian}.
Besides, this setup is physically relevant in various practical situations such as continuous-variable random quantum networks and essential to study typical behaviors of generic bosonic circuits~\cite{zhuang2019scrambling, zhang2020entanglement}. 
We provide efficient classical approximate samplers for such bosonic random circuits and compute a bound on the circuit-depth for which the probability distribution of samplers are close to the ideal distribution.
Notably, the proposed classical samplers' approximation error is exponentially small in the number of sources.
Therefore, this result can provide a lower bound of linear-optical circuits' depth to implement classically intractable outcomes.
Also, the result together with hardness of boson sampling shows a lower bound of linear-optical circuits' depth to constitute a global Haar-random unitary circuit.
Finally, we discuss how our results relate to current boson sampling experiments.


In Sec. \ref{sec:problem}, we describe the problem setup.
In Sec. \ref{sec:mapping}, we introduce a mapping between a bosonic random circuit and classical random walk.
Using the mapping, we present an efficient classical algorithm to sample outcomes of the problem with squeezed vacuum state input in Sec. \ref{sec:sq} and with single-photon state input in Sec. \ref{sec:fock}.
In Sec. \ref{sec:previous}, we compare our result to the previously known results.
Finally in Sec. \ref{sec:imply}, we discuss the implication of our results and summarize.

\begin{figure*}[t]
\includegraphics[width=500px]{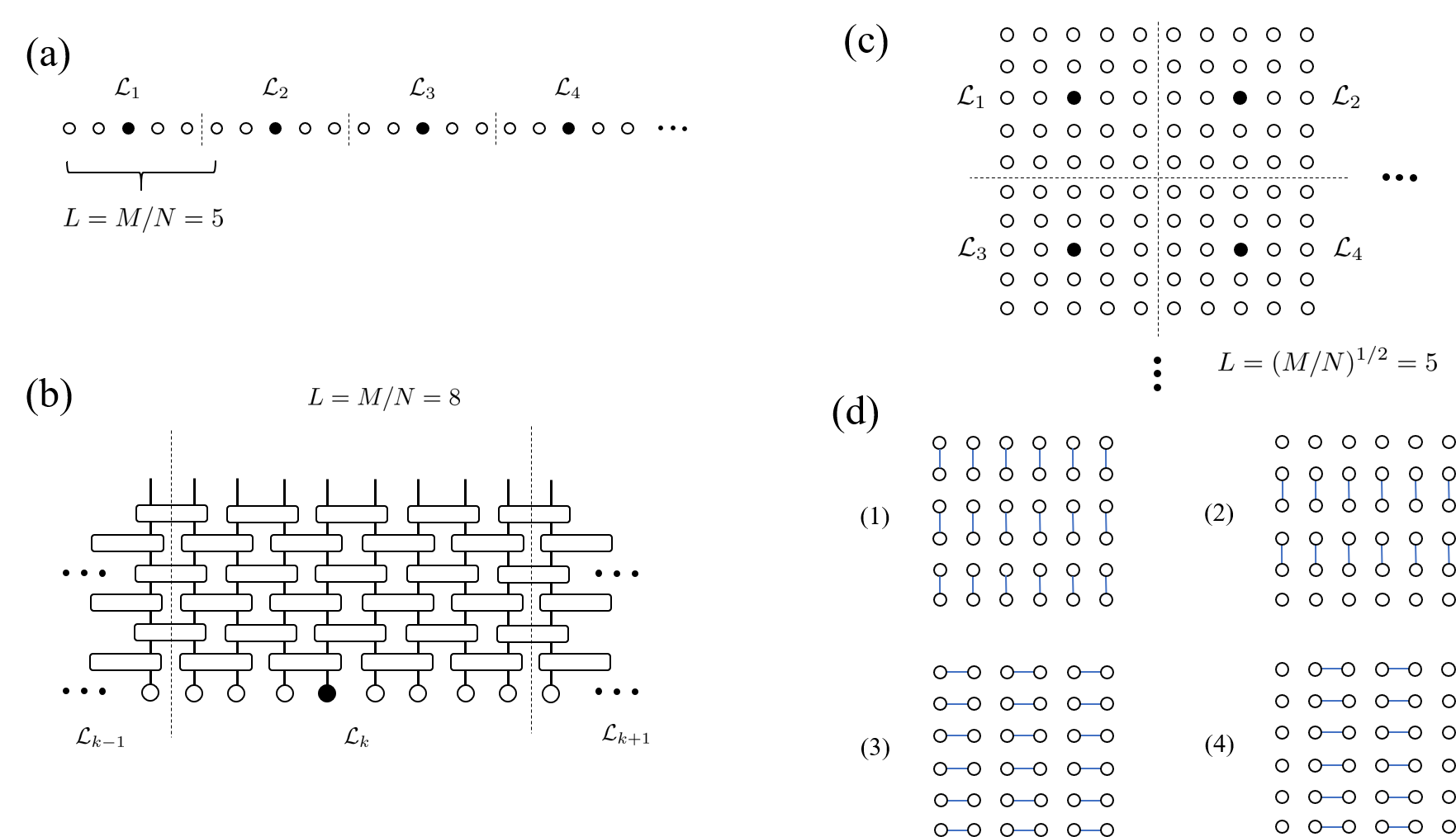}
\caption{(a) Initial state in 1d architecture. Black dots are occupied by sources and empty dots are vacuum. (b) Random beam splitter arrays in 1d. (c) Initial state in 2d architecture. (d) Random beam splitter arrays in 2d architecture. A single round consists of four steps (1)-(4) with 2 steps for each direction. The structure can be generalized for $d$-dimensional architecture, where a single round consists of $2d$ steps.}
\label{fig:circuit}
\end{figure*}

\section{Problem Setup}\label{sec:problem}
Let us consider a bosonic system where $N$ identical sources are equally distributed in $M$ bosonic modes of a $d$-dimensional lattice and evolve under random beam splitter arrays.
The problem setting is illustrated in Fig.~\ref{fig:circuit}.
Each beam splitter is characterized by an independent Haar-random unitary matrix.
More specifically, applying a beam splitter on mode $\hat{a}$ and $\hat{b}$ is described as
\begin{align}\label{eq:bs}
    \begin{pmatrix}
        \hat{a} \\ 
        \hat{b}
    \end{pmatrix}
    \rightarrow 
    \begin{pmatrix}
        \cos\theta & e^{i\phi}\sin\theta  \\ 
        -e^{-i\phi}\sin\theta  & \cos\theta
    \end{pmatrix}
    \begin{pmatrix}
        \hat{a} \\ 
        \hat{b}
    \end{pmatrix}
\end{align}
where $\theta$ and $\phi$ are independently sampled from a uniform distribution on $[0,2\pi)$.
Examples of the problem are boson sampling \cite{aaronson2011computational} and Gaussian boson sampling \cite{hamilton2017gaussian}, where identical sources correspond to a single-photon state and a squeezed vacuum state, respectively.
We denote sublattices of lattices, each of which contains a single source, as $\{\mathcal{L}_\alpha\}_{\alpha=1}^N$ satisfying $\cup_{\alpha=1}^N \mathcal{L}_\alpha=\mathcal{M}$, where $\mathcal{M}$ contains all $M$ bosonic modes in the problem, and sublattices are disjoint $\mathcal{L}_\alpha\cap\mathcal{L}_\beta=\emptyset$ for $\alpha\neq \beta$.
We also denote the set of modes having sources as $\mathcal{S}=\{s_1,\dots,s_N\}$, where $s_j$ represents a source mode in the $j$th sublattice. 
Thus, $|\mathcal{M}|=M$, $|\mathcal{S}|=N$, and $|\mathcal{L}_\alpha|=M/N$ for all $\alpha$.
Also, each sublattice is a $d$-cube with an edge length of $L=(M/N)^{1/d}$.
For simplicity, we assume $L$ to be a positive integer throughout the present paper.

Intuitively, when the depth of beam splitter arrays is not sufficiently large, correlations between different sublattices are limited.
In this case, we can treat sublattices to be independent each other without losing much information about the system.
For example, the Lieb-Robinson light cone limits the correlation between sublattices \cite{lieb1972finite}.
Using this, it has been shown that up to time $t\sim (M/N)^{1/d}$, there exists an efficient classical algorithm to sample from the output distribution of general quadratic bosonic Hamiltonian dynamics \cite{deshpande2018dynamical, maskara2019complexity}.
In this work, we focus on a circuit model of random 2-local passive transformation and find a depth in which classical simulation of sampling is efficient.
We show that for generic random circuits, the depth limit in which an efficient simulation is possible is deeper than the one found in Ref. \cite{deshpande2018dynamical, maskara2019complexity}.

On the other hand, when the depth of beam splitter arrays is large enough to implement a global Haar-random unitary, any approximate classical simulation of the relevant output probabilities is believed to be inefficient for single-mode squeezed vacuum states or single-photon states input for $M\sim N^6$ because existence of an efficient simulator leads to a collapse of polynomial hierarchy under some plausible assumptions \cite{aaronson2011computational, hamilton2017gaussian}.
Thus, two different aspects in terms of depths indicate that the depth of the random circuit plays a parameter that gives a phase transition of the relevant problem.
Considering the problem's relation to (Gaussian) boson sampling, let $M=kN^\gamma$ and assume that the measurement basis is the photon-number basis.
Let us write the transformation of mode operators by a given beam splitter circuit characterized by a unitary operator $\hat{U}$ as
\begin{align}\label{eq:unitary}
    \hat{a}_j\to \hat{U}^\dagger\hat{a}_j\hat{U}=\sum_{k=1}^M U_{jk} \hat{a}_k.
\end{align}
We note that throughout the paper, both a Haar-random unitary and a Haar-random unitary circuit mean a Haar-random matrix of $U$.

We assess the accuracy of our simulation by total variance distance between an ideal probability distribution and the classical algorithm's output probability distribution $\|\mathcal{D}-\mathcal{D}_a\|$.
Explicitly, the total variance distance between two probability distribution $P_1(n),P_2(n)$ is defined as
\begin{align}
    \|\mathcal{D}_1-\mathcal{D}_2\|\equiv \frac{1}{2}\sum_{n}|P_1(n)-P_2(n)|.
\end{align}
Here, the probability distribution follows Born's rule, $P(n)=\text{Tr}[\hat{\rho}\hat{\Pi}(n)]$ with a positive operator-valued measure (POVM) $\{\hat{\Pi}(n)\}$ with $\hat{\Pi}(n)=\otimes_{j=1}^M |n_j\rangle\langle n_j|$ satisfying $\hat{\Pi}(n)\geq 0$ and $\sum_{n} \hat{\Pi}(n)=\hat{\mathbb{1}}$. 
In our case, $\hat{\rho}$ represents the output quantum state after a given circuit, and $|n_j\rangle\langle n_j|$ represents a projector corresponding to $n_j$ photon outcome at the $j$th mode.

In this paper, we define an efficient simulator as follows:
\begin{definition}
A sampling problem is {\it easy} if there exists a classical randomized algorithm that takes as input the unitary of a given circuit and outputs a sample $x$ from a distribution $\mathcal{D}_a$ such that the total variation distance between the ideal distribution $\mathcal{D}$ and the distribution $\mathcal{D}_a$ satisfies $\|\mathcal{D}-\mathcal{D}_a\|$ $\leq$ $\mathcal{O}$(1/poly($N$)), in runtime poly($N$) (See Ref. \cite{efficient}).
We call such a classical algorithm as {\it efficient sampler} or {\it efficient simulator}.
\end{definition}
The main result of the present work is to find an efficient simulator for random passive bosonic circuits as described above.

\section{Results}
\subsection{Classical random walk behavior of random linear-optical circuits}\label{sec:mapping}
We consider two different types of input sources: (i) single-mode squeezed states and (ii) single-photon states in the problem.
The two sources are particularly important in that they constitute Gaussian boson sampling \cite{aaronson2011computational} and boson sampling \cite{hamilton2017gaussian} together with a linear-optical circuit, respectively.
Before we investigate the two cases separately, we present a key feature of random beam splitter arrays, which is the average propagation speed of input photons.
As recently studied, random beam splitter arrays can be characterized by a classical random walk \cite{zhang2020entanglement}.
More specifically, let us consider a single source at $s$th mode with depth $D$ beam splitter arrays.
Since the input states are vacuum except for the source, properties of such a system can be captured by the dynamics of the source,
\begin{align}
    \hat{a}_{k}^{(D)}=U_{k,s}^{(D)}\hat{a}_{s}^{(0)}+\text{vac}.
\end{align}
Thus, the crucial factors are $U_{k,s}^{(D)}$, where the superscript $(D)$ is used to emphasize depth $D$, and $\text{vac}$ represents contribution from vacuum states.
Using a property of a Haar-random beam splitter, one can show that after applying a random beam splitter on modes $k$ and $k+1$, the average of the absolute square of $U_{k,s}^{(D)}$ transforms as (See Appendix \ref{appendix:randomwalk})
\begin{align}
    \mathbb{E}[|U_{k,s}^{(D+1)}|^2]=
    \mathbb{E}[|U_{k+1,s}^{(D+1)}|^2]
    =\frac{\mathbb{E}[|U_{k,s}^{(D)}|^2]+\mathbb{E}[|U_{k+1,s}^{(D)}|^2]}{2},
\end{align}
which exhibits the simple random walk's property.
Here, $\mathbb{E}[\cdot]$ represents an average over random beam splitter arrays.
The configuration of beam splitter arrays in 1d architecture is straightforward, which is shown in Fig. \ref{fig:circuit} (a).
It can be generalized to 2d architecture as shown in Fig. \ref{fig:circuit} (d) and further to $d$-dimensional system, the random walk behavior of which can be described by $d$ independent random walk \cite{zhang2020entanglement}.
Now, let us consider a single source in $\alpha$th sublattice and vacuum states for other modes and define a leakage rate from the source to the outside of the sublattice as
\begin{align}
    \eta\equiv \sum_{j\in \mathcal{M}\setminus\mathcal{L}_\alpha}|U_{j,s_\alpha}|^2.
\end{align}
Using the random walk property, the following lemma shows a property of the dynamics of the source:
\begin{lemma}\label{lemma:leakage}
    In $d$-dimensional architecture, when the number of mode is $M=kN^\gamma$ and $L=(M/N)^{1/d}$, the average leakage rate $\eta$ at depth $D$ from a source to the outside of the sublattice of the source is bounded by
    \begin{align}
    \mathbb{E}[\eta]\leq 2d \exp\left(-\frac{L^2}{8D/d}\right) = 2d \exp\left(-\frac{dk^2}{8D}N^{\frac{2}{d}(\gamma-1)}\right).
    \end{align}
    Here, $\mathbb{E}[\cdot]$ is an average over a depth $D$ circuit composed of Haar-random beam splitters.
    Furthermore, for $0<c_1<dk^2/8$, an arbitrary $c_2>0$, and a depth $D\leq c_1 N^{\frac{2}{d}(\gamma-1)-c_2}$,
    \begin{align}
        \eta \leq \exp(-N^{c_2}),
    \end{align}
    with a probability $1-\delta$ over the random beam splitters, where $\delta$ is exponentially small in $N$.
\end{lemma}
\begin{proof}
    See Appendix \ref{appendix:leakage}.
\end{proof}
The lemma implies that up to a depth $D\leq c_1 L^{2-c_2}=c_1 N^{\frac{2}{d}(\gamma-1)-c_2}$ for $0<c_1<dk^2/8$ and an arbitrary $c_2>0$, only an exponentially small portion of a source can propagate to the outside of the sublattice.
The lemma plays a crucial role in both cases in the following sections.

\subsection{Squeezed vacuum state input}\label{sec:sq}
In this section, we consider a squeezed vacuum state as an input source.
Since the initial state is a Gaussian state with a zero displacement in the phase space, its covariance matrix fully characterizes the input state (See Refs. \cite{wang2007quantum, weedbrook2012gaussian, adesso2014continuous, serafini2017quantum} for more details about Gaussian states.).
The covariance matrix of a given quantum state $\hat{\rho}$ is defined as $V_{jk}=\text{Tr}[\hat{\rho} \{\hat{Q}_j,\hat{Q}_k \}]/2$ with a quadrature-operator vector $\hat{Q}\equiv (\hat{x}_1,\hat{p}_1,\dots, \hat{x}_M,\hat{p}_M)$, satisfying the canonical commutation relation $[\hat{Q}_j,\hat{Q}_k]=i\Omega_{jk}$, where
\begin{align}
	\Omega\equiv \mathbb{1}_M \otimes 
	\begin{pmatrix}
		0 & 1 \\ 
		-1 & 0
	\end{pmatrix}.
\end{align}
Note that Gaussian unitary transformations applied to a Gaussian state of a zero displacement can be equivalently characterized by applying symplectic transformations to the Gaussian state's covariance matrix, namely, $\hat{U}\hat{\rho}\hat{U}^\dagger\iff SVS^\text{T}$, where symplectic matrices conserve the canonical commutation relation, $S^\text{T}\Omega S=\Omega$.

We first begin with a vacuum state, the covariance matrix of which is given by $\mathbb{1}_{2M}/2$.
We then apply a squeezing symplectic transformation on source modes in $\mathcal{S}$, which is written as
\begin{align}
    S_\text{sq}=
    \oplus_{i=1}^M
    \begin{pmatrix}
        e^{r_i} & 0 \\
        0 & e^{-r_i}
    \end{pmatrix},
\end{align}
where $r_i=r$ for $i\in\mathcal{S}$ and $r_i=0$ otherwise.
We assume a momentum-squeezing operation with a real positive squeezing parameter $r>0$ without loss of generality.
Thus, the covariance matrix of the input state is written as
\begin{align}
    V_\text{in}=S_{\text{sq}}\frac{\mathbb{1}_{2M}}{2}S_{\text{sq}}^\text{T}=\frac{1}{2}\oplus_{i=1}^M \text{diag}(e^{2r_i},e^{-2r_i}).
\end{align}
In addition, beam splitters are also Gaussian unitary operations, so that beam splitter operations \eqref{eq:bs} between two modes can be characterized by symplectic matrices $S_\text{BS}$, which is formally written as
\begin{align}
    S_{\text{BS}}=
    \begin{pmatrix}
        \cos \theta & e^{i\phi}\sin \theta \\
        -e^{-i\phi}\sin \theta & \cos \theta
    \end{pmatrix} \otimes \mathbb{1}_2.
\end{align}
The symplectic matrix corresponding to given beam splitter arrays of depths $D$ can be efficiently computed by matrix multiplications of $2M\times 2M$ beam splitter symplectic matrices.

Now, based on the intuition that for low-depth circuits the correlation between sublattices is bounded, we give a classical algorithm which efficiently simulates a given bosonic quantum circuit.
The input and output of the algorithm are random beam splitters and a sample $(n_1,\dots, n_M)$ after the measurement on the photon-number basis, respectively.
The algorithm runs as follows:
(i) Compute the symplectic transformation of beam splitters.
(ii) For $j$ from 1 to $N$, iterate (iii) and (iv).
(iii) Initialize the $j$th source to be a squeezed vacuum state and other modes to be vacuum.
Compute the final covariance matrix for the $j$th sublattice using the symplectic matrix obtained in (i).
We keep only an $L\times L$ submatrix of the final $M\times M$ covariance on the $j$th sublattice.
(iv) Sample $(n_{j+1},\dots, n_{j+L})$ by computing the Hafnian relevant to the covariance matrix and using Monte-Carlo simulation \cite{quesada2020exact} (See Appendix \ref{appendix:GBS}).

The Monte-Carlo simulation in step (iv) is essentially a classical algorithm of Gaussian boson sampling, which generally requires an exponential amount of time in $N$.
An important property which guarantees that the simulator is efficient is that for each $j$, the relevant matrices for Hafnian computation are at most rank-2 because the system has a single source.
Remarkably, the computation cost of computing Hafnian of a $K\times K$ low rank matrix is poly$(K)$ \cite{barvinok1996two,kan2008moments,bjorklund2019faster}.
Thus, the above classical algorithm costs only polynomial time in $N$.
As a remark, a Gaussian state is a superposition of infinitely many photon states, we need to set a threshold for the total photons number for simulation, which causes a truncation error.
The probability of obtaining a large number of photons is exponentially small so that the truncation error can be properly suppressed using an appropriate threshold (See Appendix \ref{sec:threshold} for details.).
The presented algorithm can be used to simulate Gaussian boson sampling with threshold detectors, whose probability distributions are characterized by Torontonian \cite{quesada2018gaussian, zhong2020quantum}.
Basically, since threshold detectors can only discriminate between vacuum and more than or equal to a single-photon, the probability distribution for the Monte-Carlo simulation is coarse-grained.
Specifically, the probability of a detection event from threshold detectors can be obtained by computing Torontonian.

The output distribution of the algorithm is equivalent to that obtained by replacing the true covariance matrix $V_\text{out}$ with an approximated covariance matrix as follows:
\begin{align}\label{eq:approx}
	V_\text{out}=
	\begin{pmatrix}
		v_1 & x_{12} & \cdots & x_{1N} \\
		x_{21} & v_2 & \cdots & x_{2N} \\
		\vdots & \vdots & \ddots & \vdots \\
		x_{N1} & x_{N2} & \cdots & v_N \\
	\end{pmatrix}
	\to
	V_a=
	\begin{pmatrix}
		v_1' & 0 &  \cdots & 0 \\
		0 & v_2' &  \cdots & 0 \\
		\vdots & \vdots & \ddots & \vdots \\
		0 & 0 &  \cdots & v_N' \\
	\end{pmatrix},
\end{align}
where $v_j$ and $x_{jk}$ are $L\times L$ submatrices, corresponding to the $j$th sublattice's covariance matrix and correlations between $j$th and $k$th sublattices, and $v_j'$ is a $L\times L$ covariance matrix obtained in step (iii). 
Consequently, the approximated covariance matrix $V_a$ now represents a product of Gaussian states on different sublattices,  $\hat{\rho}_a=\otimes_{\alpha=1}^N \hat{\rho}_\alpha$.
One can expect that such an approximation works well when the correlation between sublattices are small enough.
Now, we find a depth where the approximation works.

Let us recall that total variance distance can be bounded by quantum infidelity $1-F$ \cite{fuchs1999cryptographic},
\begin{align}
\frac{1}{2}\sum_{n} |P(n)-P_a(n)|\leq \frac{1}{2}\| \hat{\rho}-\hat{\rho}_a \|\leq \sqrt{1-F(\hat{\rho},\hat{\rho}_a)}.
\end{align}
The upper bound of total variance distance indicates that since the outcome of a quantum state of a covariance matrix $V_a$ after a photon-number measurement can be efficiently simulated, it suffices to find how close quantum states of the covariance matrix $V_\text{out}$ and $V_a$ are in terms of quantum fidelity with respect to a quantum circuit's depth.
Quantum fidelity between two $M$-mode Gaussian states characterized by covariance matrices $V_1,V_2$, one of which is pure, can be written as \cite{spedalieri2012limit, banchi2015quantum}
\begin{align}\label{eq:gFid}
F(V_1,V_2)=\frac{1}{\sqrt{\det(V_1+V_2)}},
\end{align}
where we used a covariance matrix instead of quantum state $\hat{\rho}$ in the argument because a covariance matrix completely characterizes a quantum state in our case.
We assume that the first-moment of Gaussian states is zero throughout the paper, which is used in Eq.~\eqref{eq:gFid}.

Let us analyze the error of the approximation \eqref{eq:approx}.
The following lemma shows that the quantum infidelity between two Gaussian states characterized by $V_1$ and $V_2$ can be bounded by the Frobenius norm $\|\cdot\|$ of their difference matrix $X=V_1-V_2$.
\begin{lemma}\label{lemma:infid}
    Let $V_1$ be a covariance matrix of a Gaussian state in bosonic modes $\mathcal{M}$ obtained by applying beam splitter arrays on single-mode squeezed states of squeezing parameter $r$ and $V_2$ be a covariance matrix of a Gaussian state.
    If $X=V_1-V_2$ is small, $\|X\|\ll 1$, the quantum infidelity between the two Gaussian states is bounded by $1-F(V_1,V_2)\leq \frac{1}{2}\|X\|\sqrt{2N\cosh{4r}}+O(\|X\|^2)$.
\end{lemma}

\begin{proof}
See Appendix \ref{appendix:infid}.
\end{proof}

The lemma guarantees that if $\|X\|$ is small enough, our approximation is accurate in terms of quantum fidelity.
Intuitively, one can expect that such an error is small when a circuit is so shallow that a source from each sublattice has not propagated much to other sublattices to constitute correlations between sublattices.
The following lemma shows that the correlations are bounded by the leakage rate from a source in a sublattice to outside of the sublattice.
\begin{lemma}\label{lemma:correlation}
Let $V$ be a covariance matrix of a Gaussian state in bosonic mode $\mathcal{M}$ obtained by applying beam splitter arrays on single-mode squeezed states of squeezing parameter $r$ and $V_a$ is a covariance matrix obtained by the proposed classical algorithm. Then, the Frobenius norm of $X=V-V_a$ is bounded as
\begin{align}
\|X\|^2\leq e^{4r}N^4(\eta+2\sqrt{\eta})^2,
\end{align}
where $\eta$ is the leakage rate from a source $s$ in $\mathcal{L}_\alpha$ to other sublattices $\mathcal{M}\setminus \mathcal{L}_\alpha$.
\end{lemma}
\begin{proof}
    See Appendix \ref{appendix:correlation}.
\end{proof}

Therefore, the interaction between a sublattice and the outside of the sublattice is characterized by the leakage rate of the source in the sublattice to the outside.
Since we already have the bound of the leakage rate with respect to the number of sources $N$, modes $M$ and depth $D$ from lemma \ref{lemma:leakage}, we are ready to prove the following theorem.
\begin{theorem}
    (Easiness for squeezed-state input) Consider an $M$-mode bosonic system of $d$-dimensional architecture with $N$ number of equally distributed squeezed vacuum sources and $M=kN^\gamma$ with $\gamma>1$ and photon-counting measurements. The corresponding sampling problem is easy for depth $D\leq c_1 N^{2(\gamma-1)/d-c_2}$ with a probability $1-\delta$ over the random beam splitters. Here, $\delta$ is exponentially small in $N$ and for $0<c_1<dk^2/8$ and an arbitrary constant $c_2>0$.
\end{theorem}
\begin{proof}
Since the proposed algorithm is efficient, it suffices to find a depth in which the algorithm's error is small in $N$.
Using lemma \ref{lemma:infid} and lemma \ref{lemma:correlation}, we obtain the upper bound of total variance distance for an arbitrary POVM as
\begin{align}
\|\mathcal{D}-\mathcal{D}_a\|&\leq \left(\frac{N}{2}\cosh{4r} \|X\|^2+O(\|X\|^3) \right)^{1/4} \nonumber \\
&\leq c \eta^{1/4}N^{5/4}+O(\eta^{3/4}),
\end{align}
where $\eta$ is a leakage rate, and $c>0$ is some constant.
Lemma \ref{lemma:leakage} provides an upper bound of the leakage rate, which leads to, for large $N$,
\begin{align}
\|\mathcal{D}-\mathcal{D}_a\|&\leq c\exp\left(\frac{5}{4}\log N-\frac{N^{c_2}}{4}\right),
\end{align}
for $D\leq c_1 N^{\frac{2}{d}(\gamma-1)-c_2}$ with $0<c_1<dk^2/8$ and any $c_2>0$ with probability $1-\delta$ over the random beam splitters.
Thus, the total variance distance is exponentially small in $N$ with probability $1-\delta$.
\end{proof}

The theorem shows that our classical algorithm can simulate a quantum circuit with an exponentially small error for depth $D\leq c_1 L^{2-c_2}= c_1 N^{\frac{2}{d}(\gamma-1)-c_2}$ in polynomial time in $N$.

\subsection{Single-photon states}\label{sec:fock}
In this section, we consider single-photon sources as an initial state with the same configuration.
Single-photon sources are also important because the classical simulation of the output distribution after a Haar-random unitary circuit is proven to be hard under some plausible assumptions \cite{aaronson2011computational}.
After a passive transformation \eqref{eq:unitary}, the state evolves as
\begin{align}
    |\psi_0\rangle=\prod_{\alpha=1}^N\hat{a}^\dagger_{s_j}|0\rangle\to|\psi\rangle=\prod_{\alpha=1}^N\left(\sum_{k=1}^M U_{k,s_\alpha}\hat{a}^\dagger_{k}\right)|0\rangle.
\end{align}
We follow the method presented in Ref. \cite{deshpande2018dynamical} with a different correlation bound.
Particularly, while Ref. \cite{deshpande2018dynamical} employs the Lieb-Robinson bound \cite{lieb1972finite}, we use a correlation bound implied by Lemma  \ref{lemma:leakage}.
Let an $M$-dimensional tuple $r$ represent the initial configuration of single-photons and an $M$-dimensional tuple $t$ the final configuration of the photons from photon-number detection. For example, $r=(0,1,0,1,0)$ presents that we begin with single photon sources at the second and fourth modes when $M=5$, and $t=(0,1,1,0,0)$ presents that we detect photons at the second and third modes.
Also, we define lists $\text{in}$ and $\text{out}$ to represent positions of input and output photons, which gives $\text{in}=(2, 4)$ and $\text{out}=(2, 3)$ for the above example.
In our case, the input list is given by $\text{in}=(s_1,\dots,s_N)$.
The probability amplitude for the configuration is then written as
\begin{align}
    \phi(r,t)=\frac{1}{\sqrt{t!}}\sum_\sigma\prod_{j=1}^N U_{\text{out}_{\sigma(j)},\text{in}_j},
\end{align}
where $t!\equiv\prod_{i=1}^M t_i!$, the summation is taken over all possible permutations $\sigma$.
The probability distribution is then given by 
\begin{align}\label{eq:prob}
    P(r,t)=&|\phi(r,t)|^2
    =\frac{1}{t!}\sum_\sigma\prod_{j=1}^N |U_{\text{out}_{\sigma(j)},\text{in}_j}|^2+ \nonumber \\ 
    &\frac{1}{t!}\sum_{\sigma\neq\tau}\prod_{j=1}^N U_{\text{out}_{\sigma(j)},\text{in}_j}\left(\prod_{k=1}^N U_{\text{out}_{\tau(k)},\text{in}_k}\right)^*.
\end{align}
Our strategy is to keep the first line and ignore the second line, which is the same algorithm used in \cite{deshpande2018dynamical}.
Such an approximation can be interpreted as treating bosons as distinguishable particles.

The algorithm runs as follows:
(i) Iterate (ii) and (iii) for $j$ from 1 to $N$,
(ii) Compute the probability distribution $|U_{k,s_j}|^2$ over $1\leq k\leq M$, (iii) Sample $k$ according to the distribution and save it.
After iterating (ii) and (iii) for $j$ from 1 to $N$, we get $N$-photon outcome.
It is apparent that step (ii) and (iii) can be efficiently performed for given beam splitter arrays.
Thus, the algorithm's computational time cost is polynomial in $N$.
Hence, it suffices to show that the algorithm's accuracy is polynomially small in $N$.

Back to error analysis from Eq.~\eqref{eq:prob}, by rearranging permutations and using $|a+b|\leq |a|+|b|$, we find an upper-bound of the total variance distance (See Appendix \ref{appendix:fock} for details.)
\begin{align}\label{eq:fock_error}
    \epsilon&=\sum_t\frac{1}{2t!}\left|\sum_{\sigma\neq\tau}\prod_{k=1}^N U_{{\text{out}_{\sigma(k)}},{\text{in}_k}}\left(\prod_{l=1}^N U_{{\text{out}_{\tau(l)}},{\text{in}_l}}\right)^*\right| \\ 
    &\leq \frac{1}{2}\sum_j\sum_{\rho\neq\text{Id}}\prod_i|U_{j_i,\text{in}_i}||U_{j_i,\text{in}_{\rho(i)}}|,
\end{align}
where the sum $j$ is taken over ordered tuples $(j_1,\dots,j_N)$ with $1\leq j_l \leq M$ for all $1\leq l\leq N$, and the sum $\rho$ is taken over all possible permutations except for the identity $\text{Id}$.
Equivalently, we can express the upper-bound as
\begin{align}\label{eq:fock_upper_bound}
    \frac{1}{2}\sum_j\sum_{\rho\neq\text{Id}}\prod_i|U_{j_i,\text{in}_i}||U_{j_i,\text{in}_{\rho(i)}}|=\frac{1}{2}\sum_{\mathcal{I}_C} \prod_{i\in\mathcal{I}_C}C_i\prod_{k\in\mathcal{I}_D}D_k,
\end{align}
where the sum $\mathcal{I}_C$ is over nonempty subsets $\mathcal{I}_C$ of the indices of $\text{in}$ and $\mathcal{I}_D$ is the complement of $\mathcal{I}_C$, and
\begin{align}
    C_i&\equiv \sum_{j_i,\rho:\rho(i)\neq i}|U_{j_i,\text{in}_i}||U_{j_i,\text{in}_{\rho(i)}}|, \\ 
    D_i&\equiv \sum_{j_i,\rho:\rho(i)= i}|U_{j_i,\text{in}_i}||U_{j_i,\text{in}_{\rho(i)}}|=1.
\end{align}
Equivalently, $\mathcal{I}_D$ represents the set of fixed points.
Now, we show that $C_i$ can be bounded by the leakage rate $\eta$:
\begin{lemma}\label{lemma:c_bound}
    \begin{align}
        C_i&\leq 2\sqrt{\eta kN^{\gamma+1}}.
    \end{align}
\end{lemma}
\begin{proof}
    Appendix \ref{appendix:fock}.
\end{proof}

Finally, we show that the approximation algorithm is an efficient classical simulator for $D\leq c_1 N^{2(\gamma-1)/d-c_2}$ for $0<c_1<dk^2/8$ and an arbitrary constant $c_2>0$.
\begin{theorem}\label{ecs_fock}
    (Easiness for single-photon-state input) Consider an $M$-mode bosonic system of $d$-dimensional architecture with $N$ number of equally distributed single-photon sources and $M=kN^\gamma$ with $\gamma>1$ and photon-counting measurements. The corresponding sampling problem is easy for depth $D\leq c_1 N^{2(\gamma-1)/d-c_2}$ with a probability $1-\delta$ over the random beam splitters, with an exponentially small $\delta$ in $N$ and for $0<c_1<dk^2/8$ and an arbitrary constant $c_2>0$.
\end{theorem}
\begin{proof}
    Let $c$ be the upper-bound of $C_i$ presented in lemma \ref{lemma:c_bound}.
    From Eq.~\eqref{eq:fock_upper_bound},
    \begin{align}
    \epsilon&\leq \frac{1}{2}\sum_{l=|\mathcal{I}_C|\geq 2}^N\binom{N}{l}c^l=\frac{1}{2}[(c+1)^N-Nc-1] \\
    &\leq \exp[2\log N+(N-2)\log(1+c)+\log c^2] \\ 
    &\leq \exp[2\log N+2\sqrt{\eta k N^{\gamma+3}}+\log (4\eta kN^{\gamma+1})].
\end{align}
Here, we have used Taylor's theorem that ${[(c+1)^N-Nc-1]/2}=\binom{N}{2}(1+h)^{N-2}c^2~\text{for some}~h\in[0,c]$ for the second inequality.
Using lemma \ref{lemma:leakage}, 
\begin{align}
    \epsilon\leq \exp\left[-N^{c_2}+o(N^{c_2})\right],
\end{align}
for depth $D\leq c_1 N^{2(\gamma-1)/d-c_2}$ with a probability $1-\delta$ over random beam splitter circuits, with an exponentially small $\delta$ in $N$ and for $0<c_1<dk^2/8$ and an arbitrary constant $c_2>0$.
Hence, the total variance distance $\epsilon$ is bounded by an exponentially small number in $N$.
\end{proof}

\section{Relation to previous results}\label{sec:previous}
In this section, we compare our result with the related previous results in the literature.
First, we compare our result with the one obtained in Refs. \cite{deshpande2018dynamical, maskara2019complexity} where time evolution of single photons under spatially local quadratic bosonic Hamiltonians has been studied.
While both study a passive nearest-neighbor-interacting quadratic Hamiltonian,
a crucial difference is that we have considered {\it random} beam splitter circuits rather than time evolution under a {\it general} spatially local quadratic bosonic Hamiltonian.
Specifically, our result shows that generic linear-optical circuits up to depth $D=O(L^{2-\epsilon})=O(N^{\frac{2}{d}(\gamma-1)-\epsilon})$ are easy to classically simulate except for an exponentially small fraction of random circuits, whereas the sampling complexity phase-transition of the Hamiltonian dynamics was found at time $t=\Theta(L)=\Theta( N^{\frac{\gamma-1}{d}})$ in Refs. \cite{deshpande2018dynamical, maskara2019complexity}.
Combining the two results, there exists an exponentially small portion of linear-optical circuits between depths $D=O( N^{\frac{\gamma-1}{d}+\epsilon_1})$ and $D=O(N^{\frac{2}{d}(\gamma-1)-\epsilon_2})$ which is hard to classically simulate, where $\epsilon_{1,2}>0$ are arbitrarily small.
Such a worst-case and average-case separation can be understood by noting that except for the exponentially small portion of circuits, single-photon states or squeezed vacuum states are still localized in each sublattice.
Thus, the relevant probabilities for each outcome can be approximated by the permanent or hafnian of low-rank matrices, which allows an efficient classical simulation.
Also, since the fraction of hard circuits is exponentially small in our case, the average-to-worst-case reduction fails \cite{aaronson2011computational, bouland2019complexity, movassagh2019quantum, bouland2021noise, napp2019efficient} and the separation appears.
Finally, we note that for exact sampling of Gaussian boson sampling with threshold detectors can be efficiently simulated up to depth $D=O(\log M)$.

In addition, our result renders an important implication on boson sampling experiments.
We first note that if one uses structured beam splitter arrays \cite{reck1994experimental, clements2016optimal, russell2017direct}, depth $D=M$ is sufficient to implement a Haar-random unitary matrix $U$ by appropriately manipulating each beam splitter's transmissivity.
Nevertheless, there is a practical reason why we need to focus on generic circuits in experiments.
As boson sampling experiments' circuit size grows to demonstrate quantum supremacy, manipulating each beam splitter becomes very demanding.
Indeed, the unitary matrix of beam splitter circuits in current boson sampling experiments presents a slight deviation from a Haar-random unitary matrix \cite{wang2018toward, wang2019boson}.
In the case of boson sampling experiments using $1d$ architecture \cite{wang2017high, wang2018toward}, our easiness bound indicates that the depth of circuits has to increase as $D>c_1N^{2(\gamma-1)-c_2}\sim M^{\frac{2(\gamma-1)}{\gamma}-c_2'}$ unless beam splitters are highly well controlled.
Especially for $\gamma=2$, the depth less than a linear-scaling of $M$ is easy to simulate.
In fact, the most recent boson sampling experiment has employed $2d$ architecture to implement a global Haar-random unitary \cite{wang2019boson}.
If one scales up the size of the experiment using the same structure, one can easily check that a depth increases as $D\propto \sqrt{M}$. 
Although their configuration of beam splitters and input states is different from ours, it is still interesting to see that our easiness bound, $D\propto M^{\frac{\gamma-1}{\gamma}-c_2'}$, is smaller than the experiment for $\gamma\leq 2$.
Thus, such a configuration might be proper to increase the system size prohibiting an efficient simulation based on our results for $\gamma \leq 2$, while one needs to invent a better configuration if the number of modes is much larger than the number of sources such that $M\propto N^\gamma$ with $\gamma > 2$.

\section{Discussions and conclusions}\label{sec:imply}
In the present work, we have analyzed bosonic linear-optical random circuits with squeezed vacuum states and single-photon states and shown an efficient classical simulator for depth $D=O(N^{\frac{2}{d}(\gamma-1)-\epsilon})$ with an arbitrarily small $\epsilon>0$.
Such an easiness-depth limit shows that if $\gamma\geq 2$, for an one(two)-dimensional configuration, bosonic circuits of a depth less than quadratic (linear) can be classically simulated.
The condition $\gamma \geq 2$ is important in the context of boson sampling in that it renders collision-free outcomes \cite{aaronson2011computational}.
On the other hand, the hardness of (Gaussian) boson sampling states that an efficient sampler does not exist unless the polynomial hierarchy collapses \cite{aaronson2011computational, hamilton2017gaussian} for circuits of a global Haar-random matrix.
Our result together with the hardness of (Gaussian) boson sampling implies that in order to implement a global Haar-random unitary matrix $U$ in Eq.~\eqref{eq:unitary}, the depth of a beam splitter $M$-mode circuit needs to satisfy $D> c_1N^{\frac{2}{d}(\gamma-1)-c_2}\sim M^{\frac{2}{d}\frac{\gamma-1}{\gamma}-c_2'}$, where $\gamma\geq 2$ and $c_2'=c_2/\gamma>0$ are arbitrary.
Notably, if we take a limit of $\gamma \to \infty$, a required depth for a Haar-random unitary matrix approaches to $D=O(M^{\frac{2}{d}-\epsilon})$ where $\epsilon$ is an arbitrarily small number.
It suggests that, for example, for one(two)-dimensional architectures, one needs to constitute random beam splitters of a depth at least quadratic (linear) in the number of modes.

Our Lemma 1 has an implication on how fast quantum information propagates, which is an interesting and important topic from fundamental and practical perspectives \cite{bravyi2006lieb, jurcevic2014quasiparticle, nahum2018operator, mi2021information}.
While a Lieb-Robinson light cone \cite{lieb1972finite} gives a general bound of information propagation speed, a so-called Frobenius light cone \cite{roberts2016lieb, tran2020hierarchy} was recently shown to be tighter than the Lieb-Robinson light cone for quantum state-transfer tasks \cite{cirac1997quantum} as well as many-body quantum chaos, revealing a gap between the two light cones \cite{tran2020hierarchy}.
Basically, a Frobenius light cone is obtained by averaging a correlation function over input states while a Lieb-Robinson light cone takes into account the worst-case input state.
As the gap of the worst-case and average case was found in the literature,
our average-case study for random circuits together with the previous worst-case study for Hamiltonian dynamics in Ref. \cite{deshpande2018dynamical} also clearly reveals a gap of information propagation speeds from a different aspect.

It is worth emphasizing that diffusive dynamics leading to the sub-linearity propagation of quantum information in our system is fundamentally related to the U(1) symmetry, which conserves total excitation number \cite{zhang2020entanglement, khemani2018operator, rakovszky2018diffusive}, whereas quantum information of typical systems without symmetry spreads ballistically in space-time \cite{kim2013ballistic, nahum2018operator, xu2020accessing}.
Hence, our observation shows that extra symmetry will affect the quantum information propagation, and on the other hand also enable more efficient simulation algorithms.
It would be an interesting future work to investigate different kinds of symmetry that affect the system's sampling complexity.

Throughout the work, we have assumed an equally distributed sources and this assumption plays an important role in the proposed simulators.
Nonetheless, once a Haar-random unitary circuit is implemented, the hardness of (Gaussian) boson sampling is independent of an initial configuration, i.e., sources need not be equally distributed.
Thus, the implication of our result on the hardness of random bosonic circuits does not rely on an initial configuration.
Still, it is an interesting open question to find a depth for easiness assuming different configurations.
Also, we have focused on squeezed vacuum states and single-photon states because of the relation to (Gaussian) boson sampling.
We expect that a similar result can be derived for a different type of states because the key property of the problem was a mapping to a simple random walk presented by Lemma \ref{lemma:leakage}, although we leave such a generalization as an open question.


In conclusion, we have studied classical simulability of bosonic random circuits consisting of beam splitter arrays. We have investigated two kinds of bosonic sources: squeezed vacuum states and single-photon states.
We have provided efficient classical samplers that approximately simulate generic bosonic quantum circuits of beam splitters and shown a depth limit that the algorithms work.

\section*{acknowledgments}
We thank Owen Howell, Alireza Seif, Roozbeh Bassirian, Abhinav Deshpande for interesting and fruitful discussions.
C.O. and L.J. acknowledge support from the ARL-CDQI (W911NF-15-2-0067), ARO (W911NF-18-1-0020, W911NF-18-1-0212), ARO MURI (W911NF-16-1-0349), AFOSR MURI (FA9550-15-1-0015, FA9550-19-1-0399), DOE (DE-SC0019406), NSF (EFMA-1640959, OMA-1936118), and the Packard Foundation (2013-39273). Y. L. acknowledges National Research Foundation of Korea a grant funded by the Ministry of Science and ICT (NRF-2020M3E4A1077861) and KIAS Individual Grant
(CG073301) at Korea Institute for Advanced Study.
B.F. acknowledges support from AFOSR (YIP number FA9550-18-1-0148 and
FA9550-21-1-0008). This material is based upon work partially
supported by the National Science Foundation under Grant CCF-2044923
(CAREER).

\appendix
\section{Random walk behavior of random beam splitter arrays}\label{appendix:randomwalk}
Let us consider a single source and random beam splitter arrays acting on the source.
Let us consider two modes $\hat{a}_{k}^{(D)}$ and $\hat{a}_{k+1}^{(D)}$ at depth $D$, which are written as
\begin{align}
\hat{a}_{k}^{(D)}=U_{k,s}^{(D)}\hat{a}_{s}^{(0)}+\text{vac}, ~~\text{and}~~ \hat{a}_{k+1}^{(D)}=U_{k+1,s}^{(D)}\hat{a}_{s}^{(0)}+\text{vac},
\end{align}
where $\text{vac}$ represents the contributions from initial modes occupied by vacuum states.
After applying a beam splitter between them, the modes are transformed as
\begin{align}
    \hat{a}_{k}^{(D+1)}&=  e^{i\phi_1}(U_{k,s}^{(D)}\cos\theta+ e^{i\phi_0}  U_{k+1,s}^{(D)}\sin\theta)\hat{a}_{s,0}+\text{vac} \nonumber \\ 
    &=U_{k,s}^{(D+1)}\hat{a}_{s,0}+\text{vac},\\
    \hat{a}_{k+1}^{(D+1)}&=  e^{i\phi_2}(U_{k+1,s}^{(D)}\cos\theta- e^{-i\phi_0}  U_{k,s}^{(D)}\sin\theta)\hat{a}_{s,0}+\text{vac} \nonumber \\ 
    &=U_{k+1,s}^{(D+1)}\hat{a}_{s,0}+\text{vac},
\end{align}
where $\cos\theta$ and $\sin\theta$ represent the beam splitter's transmissivity and reflectivity, respectively, and
\begin{align}
    U_{k,s}^{(D+1)}&=e^{i\phi_1}(U_{k,s}^{(D)}\cos\theta+ e^{i\phi_0}  U_{k+1,s}^{(D)}\sin\theta),~~~ \text{and}\nonumber 
    \\ U_{k+1,s}^{(D+1)}&=e^{i\phi_2}(U_{k+1,s}^{(D)}\cos\theta- e^{-i\phi_0}  U_{k,s}^{(D)}\sin\theta).
\end{align}
After averaging the transmissivity $\cos\theta$ over a uniform distribution of $\theta\in[0,2\pi)$, and the phases $\phi_0,\phi_1$ and, $\phi_2$ over a uniform distribution of $[0,2\pi)$, we obtain
\begin{align}
    \mathbb{E}[|U_{k,s}^{(D+1)}|^2]=
    \mathbb{E}[|U_{k+1,s}^{(D+1)}|^2]
    =\frac{\mathbb{E}[|U_{k,s}^{(D)}|^2]+\mathbb{E}[|U_{k+1,s}^{(D)}|^2]}{2},
\end{align}
which shows that the transmissivity of a random beam splitter array follows a random walk behavior.

\section{Proof of Lemma \ref{lemma:leakage}}\label{appendix:leakage}
\begin{proof}
We first note that on average the random circuits can be characterized by a symmetric random walk and that the goal is to find an upper-bound of the leakage rate.
We observe that the leakage rate assuming an infinite number of modes is always larger than one with boundaries because boundaries makes the walker return to the initial lattice.
Thus, it is sufficient to find an upper-bound assuming an infinite number of modes.

For one-dimensional random walk, the probability of propagating farther than $l$ in step $t$ is given by
\begin{align}
    P\leq 2\exp\left(-\frac{l^2}{2t}\right).
\end{align}
Since we are interested in the leakage rate out of a sublattice, we set $l=L/2$ with $L=(M/N)^{1/d}$ for a $d$-dimensional case.
Taking into account the dimension of the circuit, we set $t=D/d$.
Thus, the leakage rate can be upper-bounded as
\begin{align}
    \mathbb{E}[\eta]&\leq 1-\left[1-2 \exp\left(-\frac{(L/2)^2}{2D/d}\right)\right]^d \nonumber \\ 
    &\leq 2d \exp\left(-\frac{L^2}{8D/d}\right) =2d\exp\left(-\frac{dk^2}{8D}N^{\frac{2}{d}(\gamma-1)}\right).
\end{align}
Using Markov's inequality, we obtain
\begin{align}
    P(\eta\geq a)\leq \frac{\mathbb{E}[\eta]}{a}=\frac{2d}{a} \exp\left(-\frac{dk^2}{8D}N^{\frac{2}{d}(\gamma-1)}\right).
\end{align}
Especially for depth $D\leq c_1 N^{\frac{2}{d}(\gamma-1)-c_2}$ and $a=\exp(-N^{c_2})$ with an arbitrary $c_2>0$, we find
\begin{align}
    P(\eta\geq \exp(-N^{c_2}))\leq 2d \exp\left(-\frac{dk^2}{8c_1}N^{c_2}+N^{c_2}\right),
\end{align}
which converges to zero for large $N$ when $c_1<dk^2/8$.
\end{proof}

\section{Gaussian boson sampling classical algorithm}\label{appendix:GBS}
In this Appendix, we present more details about Gaussian boson sampling and a classical algorithm for Gaussian boson sampling.
We consider $N$ number of sources in $M$ modes with random beam splitter arrays.
Let us consider a covariance matrix $\Sigma_{ij}=\text{Tr}[\hat{\rho}\{\hat{\xi}_i,\hat{\xi}_j\}]/2$ of a Gaussian state $\hat{\rho}$ with $\hat{\xi}=(\hat{a}_1,\dots,\hat{a}_M,\hat{a}_1^\dagger,\dots,\hat{a}_M^\dagger)$ to follow a notational convention used in Ref. \cite{hamilton2017gaussian,quesada2020exact}.
Note that a covariance $\Sigma$ can be easily obtained by a covariance matrix $V$, which we have used in the main text.

In the case of general Gaussian input states of a covariance matrix $\Sigma$ and a zero displacement, the probability of each outcome $(n_1,n_2,\dots,n_M)$ obtained by the measurement in photon-number basis is given by \cite{hamilton2017gaussian}
\begin{align}
    P(n_1,n_2,\dots,n_M)=\frac{1}{\sqrt{\text{det}(\Sigma+\mathbb{1}_{2M}/2)}}\frac{\text{Haf}(A_n)}{n_1!n_2!\cdots n_M!},
\end{align}
where
\begin{align}
    A=
    Y_M[\mathbb{1}_{2M}-(\Sigma+\mathbb{1}_{2M}/2)^{-1}], ~~~
    Y_m=
    \begin{pmatrix}
    0 & \mathbb{1}_M \\
    \mathbb{1}_M & 0
    \end{pmatrix}.
\end{align}
Here, $A_n$ is a matrix obtained by repeating the $j$th and $(j+M)$th row and column of $A$ for $n_j$ times for $1\leq j \leq M$.

Now, let $\Sigma^{(k)}$ be the reduced covariance matrix of the first $k$ modes.
From the reduced covariance matrix, one can constitute matrices $O^{(k)}\equiv \mathbb{1}_{2k}-(Q^{(k)})^{-1}$, and $A^{(k)}\equiv Y^{(k)}O^{(k)}$, and the latter gives a marginal probability distribution on the first $k$ modes as
\begin{align}\label{eq:GBS_prob}
    P(n_1,\dots,n_k)=\frac{1}{\sqrt{\det (\Sigma^{(k)}+\mathbb{1}_{2k}/2)}}\frac{\text{Haf}(A_n^{(k)})}{n_1!\cdots n_k!},
\end{align}
where $A_n^{(k)}$ is the matrix obtained by repeating rows and columns $i$ and $i+M$ of the matrix $A$ $n_i$ times for $1\leq i \leq k$.

By directly computing Hafnian of $A_n^{(k)}$, one may use a Monte-Carlo method to sample outcomes as follows \cite{quesada2020exact}:
First, we compute a marginal probability distribution $P(n_1)$ for $0\leq n_1\leq n_\text{max}$ and sample $n_1^*$ based on the distribution.
Since $n_j$ can be infinitely large in principle, we choose an upper-threshold of $n_j$ carefully (See Appendix \ref{sec:threshold}).
After sampling $n_1^*$, we compute $P(n_1^*,n_2)$ for $0\leq n_2\leq n_\text{max}$, sample $n_2^*$ according to the condition probability distribution
\begin{align}
    P(n_2|n_1^*)=\frac{P(n_1^*,n_2)}{P(n_1^*)}.
\end{align}
For $k$th step of the above procedure, we have $(n_1^*,\dots, n_{k-1}^*)$, so that we compute $P(n_1^*,\dots,n_{k-1}^*,n_k)$ and sample $n_k$ according to the conditional probability distribution
\begin{align}
    P(n_k|n_1^*,\dots,n_{k-1}^*)=\frac{P(n_1^*,\dots,n_{k-1}^*,n_k)}{P(n_1^*,\dots,n_{k-1}^*)}.
\end{align}
We iterate this procedure until we get $(n_1,\dots, n_{M})$.
Since the sampling procedure can be efficiently executed, the dominant computational cost is to compute the Hafnian of the relevant matrices.
when we compute the probability distribution $P(n_1)$, one needs to compute the Hafnian of a matrix.
In general, the computational cost to compute the Hafnian of a complex $2k\times 2k$ matrix is $O(k^3 2^{k})$\cite{bjorklund2019faster}, and thus the computational cost to simulate Gaussian boson sampling is $O(M n_\text{max}^3 2^{n_\text{max}})$.
However, if the rank $R$ of a matrix is low, the computational cost to compute its Hafnian can be significantly reduced as $\binom{2M+R-1}{R-1}\text{poly}(2M)$ \cite{kan2008moments,bjorklund2019faster}.
Using this, we now show that classical simulation of beam splitter circuits beginning with a single squeezed vacuum state can be efficient.
Note that if we assume threshold detectors instead of photon-counting detectors \cite{quesada2018gaussian}, the above probability distributions are coarse-grained such that the outcomes are vacuum or more than or equal to a single-photon.
The probability for the latter is obtained by computing Torontonian.

Especially when the input state is comprised of squeezed vacuum states of real squeezing parameters $r_j$, we can simplify the matrix as, where 
\begin{align}\label{eq:b_mat}
    A=B\oplus B^*, ~~~ \text{where} ~~~ B=K(\oplus_{j=1}^M \tanh r_j)K^\text{T}.
\end{align}
Here, $K$ is an $M\times M$ matrix describing beam splitter arrays, and $r_j$ is a squeezing parameter, which is non-zero only for the sources.
In this case, the covariance matrix $\Sigma$ of the output state is obtained by
\begin{align}
    \Sigma=\frac{1}{2}
    \begin{pmatrix}
        K & 0 \\
        0 & K^*
    \end{pmatrix}
    SS^\dagger
    \begin{pmatrix}
        K^\dagger & 0 \\
        0 & K^\text{T}
    \end{pmatrix},
\end{align}
where
\begin{align}
    S=
    \begin{pmatrix}
        \oplus_{j=1}^M \cosh{r_j} & \oplus_{j=1}^M \sinh{r_j} \\ 
        \oplus_{j=1}^M \sinh{r_j} & \oplus_{j=1}^M \cosh{r_j}
    \end{pmatrix}
\end{align}
represents the squeezing symplectic matrix.
Here, we assume a single squeezed state input, so that $r_1>0$ and $r_j=0$ for $j>1$ without loss of generality.
From Eq.~\eqref{eq:b_mat}, one can immediately see that $A_n$ is a rank-2 matrix for any $(n_1,\dots,n_M)$.
Now, it suffices to show that $A_n^{(k)}$ is also rank-2.
It is known that a reduced covariance matrix from a single-source covariance matrix can always be written as \cite{botero2003modewise}
\begin{align}
    \Sigma^{(k)}&=\frac{1}{2}    
    \begin{pmatrix}
        L & 0 \\
        0 & L^*
    \end{pmatrix}
    S^{(k)} D^{(k)}S^{(k)\dagger}
    \begin{pmatrix}
        L^\dagger & 0 \\
        0 & L^\text{T}
    \end{pmatrix}, \\ 
    D^{(k)}&=\text{diag}(\nu_1,\dots,\nu_k)\oplus\text{diag}(\nu_1,\dots,\nu_k), \\
    S^{(k)}&=
    \begin{pmatrix}
        \oplus_{j=1}^k \cosh{r_j^{(k)}} & \oplus_{j=1}^k \sinh{r_j^{(k)}} \\ 
        \oplus_{j=1}^k \sinh{r_j^{(k)}} & \oplus_{j=1}^k \cosh{r_j^{(k)}}
    \end{pmatrix},
\end{align}
where $\nu_1>1$, and $\nu_j=1$ for $2\leq j \leq k$, and $r_1^{(k)}>0$ and $r_j^{(k)}=0$ for $2\leq j \leq k$.
In other words, the reduced state is a product of a thermal state and vacuum states followed by a squeezing operation and beam splitter arrays.
One can easily check that $A_n^{(k)}$ is also rank-2.
Thus, in this case, the computation of Hafnian is polynomial in $M$, so that the classical simulation using the Monte-Carlo simulation is efficient.

\section{Error of photon-number truncation}\label{sec:threshold}
In this section, we analyze the influence of photon-number truncation in squeezed states.
When we have $N$ squeezed vacuum states and measure them in photon-number basis, 
the probability to generate a total of $k$ photon pair events (2$k$ photons) is given by the negative binomial distribution \cite{hamilton2017gaussian, kruse2019detailed}
\begin{align}
    P_N(k)=\binom{\frac{N}{2}+k-1}{k}\text{sech}^Nr\tanh^{2k}r.
\end{align}
Note that beam splitter arrays do not change the probability distribution of total photon numbers.
The tail probability of the negative binomial distribution is given by \cite{brown2011wasted}
\begin{align}
    \text{Pr}(k> \alpha N \text{sech}^2 r)\leq \exp\left[\frac{-\alpha N(1-1/\alpha)^2}{2}\right].
\end{align}
While increasing a constant $\alpha$ decreases the truncation error exponentially, the order of the complexity $O(M n_\text{max}^3 2^{n_\text{max}})$ does not change.
Any accuracy can be achieved by increasing $\alpha$ with a constant factor which reduces error exponentially.
In order to make the truncation error to be smaller than $\epsilon$ for sufficiently large $\alpha$, we may choose $\alpha N \text{sech}^2r= 2\text{sech}^2r\log(1/\epsilon)\equiv n_\text{max}$.

\section{Proof of Lemma \ref{lemma:infid}}\label{appendix:infid}
\begin{proof}
Note that the covariance matrix $V_1$ can be decomposed as $V_1=S(\mathbb{1}_{2M}/2)S^\text{T}$ by a symplectic matrix $S$ satisfying $S\Omega S^\text{T}=\Omega$ and that the symplectic matrix can be decomposed as $S=O S_\text{sq}$, where $O$ represents a symplectic matrix corresponding to beam splitter arrays and $S_\text{sq}$ represents squeezing operators to generate squeezed vacuum sources.
If $V_1=V_2+X$ and $\|X\|\ll 1$, we have
\begin{align}
\det(V_1+V_2)&=\det(2V_1-X)=\det(\mathbb{1}_{2M}-S^{-1} X S^{-\text{T}}) \nonumber \\ 
&\leq \left(1+\frac{1}{2M}|\text{Tr}[S^{-1} X S^{-\text{T}}|]\right)^{2M} \nonumber \\
&=\sum_{k=0}^{2M} \binom{2M}{k}\left(\frac{1}{2M}|\text{Tr}[S^{-1} X S^{-\text{T}}]|\right)^k \nonumber \\ 
&\leq \sum_{k=0}^{2M} \binom{2M}{k}\left(\frac{1}{2M}\|X\|\sqrt{\text{Tr}(S^{-\text{T}}S^{-1})^2}\right)^k \nonumber \\ 
&=1+\|X\|\sqrt{2N\cosh{4r}}+O(\|X\|^2),
\end{align}
where for the first inequality, we have used the upper-bound of determinant and that
\begin{align}
\text{Tr}[S^{-1}X S^{-\text{T}}]&=\text{Tr}[S^{-1}(V_1-V_2) S^{-\text{T}}] \nonumber \\ 
&=\text{Tr}[\mathbb{1}_{2M}/2-S^{-1}V_2S^{-\text{T}}] \nonumber \\ 
&=\text{Tr}[\mathbb{1}_{2M}/2-\tilde{V}_2]\leq 0.
\end{align}
The Cauchy-Schwarz inequality has been used for the second inequality.
Then, the quantum fidelity between Gaussian states with covariance matrices $V_1$ and $V_2$ is approximated as
\begin{align}
1-F(V_1,V_2)&=1-\frac{1}{\sqrt{\det(V_1+V_2)}} \nonumber \\ 
&\leq \frac{1}{2}\|X\|\sqrt{2N\cosh{4r}}+O(\|X\|^2).
\end{align}
\end{proof}

\begin{widetext}
\section{Proof of lemma \ref{lemma:correlation}}\label{appendix:correlation}
In this section, we compute $\|X\|^2$ and prove Lemma \ref{lemma:correlation}.
We first note that $X$ is the difference matrix between two covariance matrices $V_\text{out}$ and $V_a$.
\begin{align}\label{eq:x_sec}
    \|X\|^2
    &=\sum_{\alpha=1}^{N}\sum_{i,j\in \mathcal{L}_\alpha}X_{ij}^2+
    \sum_{1\leq \alpha\neq \beta\leq N}\sum_{i\in \mathcal{L}_\alpha}\sum_{j\in \mathcal{L}_\beta}X_{ij}^2
    \nonumber \\
    &=\frac{1}{4}\sum_{\alpha=1}^{N}\sum_{i\in \mathcal{L}_\alpha}(\delta \langle \{\hat{x}_{i},\hat{x}_{j}\}\rangle^2+\delta \langle \{\hat{p}_{i},\hat{p}_{j}\}\rangle^2+2\delta \langle \{\hat{x}_{i},\hat{p}_{j}\}\rangle^2) \nonumber \\ 
    &+\frac{1}{4}\sum_{1\leq \alpha\neq \beta\leq N}\sum_{i\in \mathcal{L}_\alpha}\sum_{j\in \mathcal{L}_\beta}\left(\langle\{\hat{x}_i,\hat{x}_j\}\rangle^2+\langle\{\hat{p}_i,\hat{p}_j\}\rangle^2+2\langle\{\hat{x}_i,\hat{p}_j\}\rangle^2\right),
\end{align}
where $\delta \langle \{\hat{O}_1,\hat{O}_2\}\rangle\equiv \langle \{\hat{O}_1,\hat{O}_2\}\rangle-\langle\{ \hat{O}_1,\hat{O}_2\}\rangle_a$ with $\langle \cdot \rangle$ and $\langle \cdot \rangle_a$ representing the expectation value for an exact and an approximate covariance matrix, respectively.
We compute the elements one by one.
First, we focus on the first terms, where $i,j$ are in the same sublattices.
Using
\begin{align}
    \hat{x}_i=\frac{1}{\sqrt{2}}(U_{ik}\hat{a}_k+U_{ik}^*\hat{a}_k^\dagger), ~~~\text{and}~~~ \hat{p}_i=\frac{i}{\sqrt{2}}(U_{ik}^*\hat{a}_k^\dagger-U_{ik}\hat{a}_k),
\end{align}
with $\hat{a}_k$ being an annihilation operator before beam splitter arrays,
we obtain
\begin{align}
    \langle \{\hat{x}_i,\hat{x}_j\}\rangle
    &=\sum_{k=1}^M\langle U_{ik}U_{jk}\hat{a}_k^2+U_{ik}^{*}U_{jk}^{*}\hat{a}_k^{\dagger2}+U_{ik}U_{jk}^*\hat{a}_k\hat{a}_k^\dagger+U_{ik}^*U_{jk}\hat{a}_k^\dagger\hat{a}_k\rangle, \\ 
    \langle \{\hat{p}_i,\hat{p}_j\}\rangle
    &=-\sum_{k=1}^M\langle U_{ik}U_{jk}\hat{a}_k^2+U_{ik}^*U_{jk}^*\hat{a}_k^{\dagger2}-U_{ik}U_{jk}^*\hat{a}_k\hat{a}_k^\dagger-U_{ik}^*U_{jk}\hat{a}_k^\dagger\hat{a}_k\rangle, \\ 
    \langle \{\hat{x}_i,\hat{p}_j\}\rangle
    &=i\sum_{k=1}^M\langle -U_{ik}U_{jk}\hat{a}_k^2+U_{ik}^*U_{jk}^*\hat{a}_k^{\dagger2}\rangle.
\end{align}
Here, note that the approximate covariance matrix is obtained by assuming other sublattices to be vacuum at the beginning.
Thus, for a fixed $\alpha$, the difference between the exact and approximate covariance matrices is that for the approximate covariance matrix, $\langle \hat{a}_k^2\rangle_a=\langle \hat{a}_k^{\dagger2}\rangle_a=\langle \hat{a}_k^\dagger\hat{a}_k\rangle_a=0$ and $\langle \hat{a}_k\hat{a}_k^\dagger\rangle_a=1$ for $k\not\in \mathcal{L}_\alpha$ but for the exact covariance matrix, $\langle \hat{a}_k^2\rangle=\langle \hat{a}_k^{\dagger2}\rangle=\cosh{r}\sinh{r}$, $\langle \hat{a}_k^\dagger\hat{a}_k\rangle=\cosh^2 r$ and $\langle \hat{a}_k\hat{a}_k^\dagger\rangle=\sinh^2{r}$ for $k\in \mathcal{S}$.
Also, $\langle \hat{a}_k^2\rangle=\langle \hat{a}_k^2\rangle_a,\langle \hat{a}_k^{\dagger2}\rangle=\langle \hat{a}_k^{\dagger2}\rangle_a,\langle \hat{a}_k^\dagger\hat{a}_k\rangle=\langle \hat{a}_k^\dagger\hat{a}_k\rangle_a=0$ and $\langle \hat{a}_k\hat{a}_k^\dagger\rangle=\langle \hat{a}_k\hat{a}_k^\dagger\rangle_a$ for $k\in \mathcal{L}_\alpha$.
Therefore, the approximation error for a fixed $\alpha$ is written as
\begin{align}
    \left|\delta \langle \{\hat{x}_i,\hat{x}_j\}\rangle\right|
    &=\left|\sum_{k\in\mathcal{M}\setminus\mathcal{L}_\alpha}\langle U_{ik}U_{jk}\hat{a}_k^2+U_{ik}^{*}U_{jk}^{*}\hat{a}_k^{\dagger2}+U_{ik}U_{jk}^*\hat{a}_k\hat{a}_k^\dagger+U_{ik}^*U_{jk}\hat{a}_k^\dagger\hat{a}_k\rangle-\sum_{k\in\mathcal{M}\setminus\mathcal{L}_\alpha}U_{ik}U^*_{jk}\right| \\
    &=\left|\sum_{k\in \mathcal{S}\setminus\{s_\alpha\}} \left(U_{ik}U_{jk}\cosh{r}\sinh{r}+U_{ik}^{*}U_{jk}^{*}\cosh{r}\sinh{r}+U_{ik}U_{jk}^*\sinh^2r+U_{ik}^*U_{jk}\sinh^2{r}\right)\right| \\ 
    &\leq 2\sum_{k\in \mathcal{S}\setminus\{s_\alpha\}} |U_{ik}||U_{jk}|e^r\sinh{r}
    \leq e^{2r}\sum_{k\in \mathcal{S}\setminus\{s_\alpha\}} |U_{ik}||U_{jk}|.
\end{align}
Now, we compute the sum of the absolute value for different $i$ and $j$,
\begin{align}
    \sum_{\alpha=1}^N\sum_{i,j\in \mathcal{L}_\alpha}|\delta\langle\{\hat{x}_i,\hat{x}_j\}\rangle|^2
    &\leq e^{4r}\sum_{\alpha=1}^N\sum_{i,j\in \mathcal{L}_\alpha}\sum_{k,l\in \mathcal{S}\setminus\{s_\alpha\}}  |U_{ik}||U_{jk}||U_{il}||U_{jl}| \\ 
    &\leq e^{4r}\sum_{\alpha=1}^N\sum_{k,l\in \mathcal{S}\setminus\{s_\alpha\}}\sqrt{\sum_{i\in\mathcal{L}_\alpha}|U_{ik}|^2\sum_{i\in\mathcal{L}_\alpha}|U_{il}|^2\sum_{j\in\mathcal{L}_\alpha}|U_{jk}|^2\sum_{j\in\mathcal{L}_\alpha}|U_{jl}|^2} \\ 
    &\leq e^{4r}\sum_{\alpha=1}^N\sum_{k,l\in \mathcal{S}\setminus\{s_\alpha\}}\eta^2=e^{4r}N(N-1)^2\eta^2.
\end{align}
Here, we have used the Cauchy-Schwarz inequality and defined $\eta$ to be a leakage rate from a source to other sublattices.
Following the same procedure, one can find the same upper-bound for momentum quadratures,
\begin{align}
    \sum_{\alpha=1}^N\sum_{i,j\in \mathcal{L}_\alpha}|\delta\langle\{\hat{p}_i,\hat{p}_j\}\rangle|^2\leq e^{4r}N(N-1)^2\eta^2.
\end{align}

Similarly, we get
\begin{align}
    \left|\delta \langle \{\hat{x}_i,\hat{p}_j\}\rangle\right|
    &=\left|\sum_{k\in\mathcal{M}\setminus\mathcal{L}_\alpha}\langle -U_{ik}U_{jk}\hat{a}_k^2+U_{ik}^*U_{jk}^*\hat{a}_k^{\dagger2}\rangle\right|
    =\left|\sum_{k\in\mathcal{M}\setminus\mathcal{L}_\alpha}(-U_{ik}U_{jk}+U_{ik}^*U_{jk}^*)\cosh{r}\sinh{r}\right| \\
    &\leq 2\sum_{k\in \mathcal{S}\setminus\{s_\alpha\}} |U_{ik}||U_{jk}|\cosh{r}\sinh{r} 
    \leq \frac{e^{2r}}{2}\sum_{k\in \mathcal{S}\setminus\{s_\alpha\}} |U_{ik}||U_{jk}|.
\end{align}

Again, using the Cauchy-Schwarz inequality, we obtain
\begin{align}
    \sum_{\alpha=1}^N\sum_{i,j\in \mathcal{L}_\alpha}|\delta \langle \{\hat{x}_i,\hat{p}_j\}\rangle|^2
    \leq \frac{e^{4r}}{4}N(N-1)^2\eta^2.
\end{align}

Now, we obtain the upper-bound of the second terms in Eq.~\eqref{eq:x_sec}, where $i$ and $j$ are in different sublattices
Here, we approximate correlations as zero, which leads to
\begin{align}
    |\delta\langle \{\hat{x}_i,\hat{x}_j\}\rangle|
    &=\left|\sum_{k=1}^M\langle U_{ik}U_{jk}\hat{a}_k^2+U_{ik}^{*}U_{jk}^{*}\hat{a}_k^{\dagger2}+U_{ik}U_{jk}^*\hat{a}_k\hat{a}_k^\dagger+U_{ik}^*U_{jk}\hat{a}_k^\dagger\hat{a}_k\rangle\right|\\
    &=\left|\sum_{k\in \mathcal{S}}\left( U_{ik}U_{jk}\cosh{r}\sinh{r}+U_{ik}^{*}U_{jk}^{*}\cosh{r}\sinh{r}+U_{ik}U_{jk}^*\cosh^2r+U_{ik}^*U_{jk}\sinh^2r\right)+\sum_{k\not\in \mathcal{S}}U_{ik}U_{jk}^*\right|\\
    &=\left|\sum_{k\in \mathcal{S}}\left( U_{ik}U_{jk}\cosh{r}\sinh{r}+U_{ik}^{*}U_{jk}^{*}\cosh{r}\sinh{r}+U_{ik}U_{jk}^*\sinh^2r+U_{ik}^*U_{jk}\sinh^2r\right)\right|\\
    &\leq 2e^{r}\sinh{r}\sum_{k\in \mathcal{S}} |U_{ik}||U_{jk}|\leq e^{2r}\sum_{k\in S} |U_{ik}||U_{jk}|.
\end{align}
After taking summation over different pairs of $\alpha$ and $\beta$, we obtain
\begin{align}
    \sum_{1\leq \alpha\neq \beta\leq N}\sum_{i\in \mathcal{L}_\alpha}\sum_{j\in\mathcal{L}_\beta} |\delta\langle \{\hat{x}_i,\hat{x}_j\}\rangle|^2 
    &\leq e^{4r}\sum_{1\leq \alpha\neq \beta\leq N}\sum_{i\in \mathcal{L}_\alpha}\sum_{j\in\mathcal{L}_\beta}\sum_{k,l\in \mathcal{S}} |U_{ik}||U_{jk}||U_{il}||U_{jl}| \\ 
    &\leq e^{4r}\sum_{1\leq \alpha\neq \beta\leq N}\sum_{k,l\in \mathcal{S}}\sqrt{\sum_{i\in \mathcal{L}_\alpha} |U_{ik}|^2\sum_{i\in \mathcal{L}_\alpha}|U_{il}|^2\sum_{j\in\mathcal{L}_\beta}|U_{jk}|^2\sum_{j\in\mathcal{L}_\beta}|U_{jl}|^2} \\ 
    &=e^{4r}\sum_{1\leq \alpha\neq \beta\leq N}\sum_{k\in \mathcal{S}}\sqrt{\sum_{i\in \mathcal{L}_\alpha} |U_{ik}|^2\sum_{j\in\mathcal{L}_\beta}|U_{jl}|^2}\sum_{l\in \mathcal{S}}\sqrt{\sum_{i\in \mathcal{L}_\alpha}|U_{il}|^2\sum_{j\in\mathcal{L}_\beta}|U_{jk}|^2}.
\end{align}
Here,    
\begin{align}
    \sum_{k\in \mathcal{S}}\sqrt{\sum_{i\in \mathcal{L}_\alpha} |U_{ik}|^2\sum_{j\in\mathcal{L}_\beta}|U_{jk}|^2}
    &=\sum_{k\in \mathcal{S}\setminus\{s_\alpha,s_\beta\}}\sqrt{\sum_{i\in \mathcal{L}_\alpha} |U_{ik}|^2\sum_{j\in\mathcal{L}_\beta}|U_{jk}|^2}
    +\sqrt{\sum_{i\in \mathcal{L}_\alpha} |U_{is_\alpha}|^2\sum_{j\in\mathcal{L}_\beta}|U_{js_\alpha}|^2} \\
    &+\sqrt{\sum_{i\in \mathcal{L}_\alpha} |U_{is_\beta}|^2\sum_{j\in\mathcal{L}_\beta}|U_{js_\beta}|^2} \\ 
    &\leq \sum_{k\in \mathcal{S}\setminus\{s_\alpha,s_\beta\}} (\eta+2\sqrt{\eta})=(N-2)(\eta+2\sqrt{\eta}).
\end{align}
Thus, the sum of the approximation errors for different pairs of $\alpha$ and $\beta$ is bounded as
\begin{align}
    \sum_{1\leq \alpha\neq \beta\leq N}\sum_{i\in \mathcal{L}_\alpha}\sum_{j\in\mathcal{L}_\beta} |\delta\langle \{\hat{x}_i,\hat{x}_j\}\rangle|^2 \leq e^{4r}N(N-1)(N-2)^2(\eta+2\sqrt{\eta})^2.
\end{align}
Again, the same upper-bound can be found for momentum quadratures,
\begin{align}
    \sum_{1\leq \alpha\neq \beta\leq N}\sum_{i\in \mathcal{L}_\alpha}\sum_{j\in\mathcal{L}_\beta} |\delta\langle \{\hat{p}_i,\hat{p}_j\}\rangle|^2 \leq e^{4r}N(N-1)(N-2)^2(\eta+2\sqrt{\eta})^2.
\end{align}

Finally, the correlation of position and momentum quadratures is bounded as
\begin{align}
    |\delta\langle \{\hat{x}_i,\hat{p}_j\}\rangle|
    &=\left|\sum_{k=1}^M\langle -U_{ik}U_{jk}\hat{a}_k^2+U_{ik}^*U_{jk}^*\hat{a}_k^{\dagger2}\rangle \right|
    =\cosh{r}\sinh{r}\left|\sum_{k\in S}(-U_{ik}U_{jk}+U_{ik}^*U_{jk}^*)\right| \\ 
    &\leq 2\cosh{r}\sinh{r}\sum_{k\in \mathcal{S}}|U_{ik}||U_{jk}|\leq \frac{e^{2r}}{2}\sum_{k\in \mathcal{S}}|U_{ik}||U_{jk}|,
\end{align}
which leads to the sum of the errors as
\begin{align}
    \sum_{1\leq \alpha\neq \beta\leq N}\sum_{i\in \mathcal{L}_\alpha}\sum_{j\in\mathcal{L}_\beta} |\delta\langle \{\hat{x}_i,\hat{p}_j\}\rangle|^2 
    &\leq \frac{e^{4r}}{4}N(N-1)(N-2)^2(\eta+2\sqrt{\eta})^2.
\end{align}

Hence, we prove lemma \ref{lemma:correlation}
\begin{align}
    \|X\|^2\leq e^{4r}\left[N(N-1)^2\eta^2+N(N-1)(N-2)^2(\eta+2\sqrt{\eta})^2\right]\leq e^{4r}N^4(\eta+2\sqrt{\eta})^2.
\end{align}
\clearpage
\end{widetext}

\section{Single-photon state}\label{appendix:fock}
In this section, we analyze the error of an approximation by treating single-photons as distinguishable particles.
From Eq.~\eqref{eq:fock_error},
\begin{align}
    \epsilon&=\sum_t\frac{1}{2t!}\left|\sum_{\sigma\neq\tau}\prod_{k=1}^N U_{{\text{out}_{\sigma(k)}},{\text{in}_k}}\left(\prod_{l=1}^N U_{{\text{out}_{\tau(l)}},{\text{in}_l}}\right)^*\right| \\ 
    &\leq \frac{1}{2}\sum_{t}\sum_{\sigma\neq\tau} \left|\prod_{k=1}^N U_{{\text{out}_{\sigma(k)}},{\text{in}_k}}\right|\left|\prod_{l=1}^N U_{{\text{out}_{\tau(l)}},{\text{in}_l}}\right| \\ 
    &=\frac{1}{2}\sum_{t}\sum_{\sigma,\rho} \prod_{k=1}^N|U_{\text{out}_{\sigma(k)},\text{in}_k}U_{\text{out}_{\sigma(k)},\text{in}_{\rho(k)}}|\\
    &= \frac{1}{2}\sum_j\sum_{\rho\neq\text{Id}}\prod_i|U_{j_i,\text{in}_i}||U_{j_i,\text{in}_{\rho(i)}}|,
\end{align}
where we have used $|a+b|\leq |a|+|b|$ for the first inequality and rearranged the summation in the last equality.
Here, the sum $j$ is taken over ordered tuples $(j_1,\dots,j_N)$ with $1\leq j_l \leq M$ for all $1\leq l\leq N$, and the sum $\rho$ is taken over all possible permutations except for the identity $\text{Id}$.

Now, we prove lemma \ref{lemma:c_bound}.
Let us find an upper-bound of $C_i$:
\begin{align}
    C_i&=\sum_{j_i,\rho:\rho(i)\neq i}|U_{j_i,i}||U_{j_i,\rho(i)}|
    =\sum_{j:|i-j|\leq L/2}\sum_{k\in\text{in},k\neq i}|U_{j,i}||U_{j,k}|  \nonumber \\
    &~~~~~~~~~~~~~~~~~~~~~~+\sum_{j:|i-j|> L/2}\sum_{k\in\text{in},k\neq i}|U_{j,i}||U_{j,k}|.
\end{align}
Here, a distance $|i-j|$ is $l_\infty$ distance between $i$th position and $j$th position.
Consider the first term of $C_i$:
\begin{align}
    &\sum_{j:|i-j|\leq L/2}|U_{j,i}|\sum_{k\in\text{in},k\neq i}|U_{j,k}| \\
    \leq& \sum_{j:|i-j|\leq L/2}|U_{j,i}|\sqrt{(N-1)\sum_{k\in\text{in},k\neq i}|U_{j,k}|^2} \\
    \leq& \sum_{j:|i-j|\leq L/2}|U_{j,i}|\sqrt{(N-1)\eta} \\ 
    \leq& \sqrt{\sum_{j:|i-j|\leq L/2}|U_{j,i}|^2}\sqrt{\eta M(N-1)/N} \\
    \leq& \sqrt{\eta M (N-1)/N},
\end{align}
where we have used the Cauchy-Schwarz inequality for the first and second inequalities.

Now, the second term:
\begin{align}
    &\sum_{j:|i-j|> L/2}|U_{j,i}|\sum_{k\in\text{in},k\neq i}|U_{j,k}| \\ 
    \leq& \sum_{j:|i-j|> L/2}|U_{j,i}|\sqrt{N-1} \\ 
    \leq& \sqrt{(M-N)(N-1)\sum_{j:|i-j|> L/2}|U_{j,i}|^2} \\ 
    \leq& \sqrt{\eta (M-N)(N-1)},
\end{align}
where we have used the Cauchy-Schwarz inequality for the first and the second inequalities.
Thus,
\begin{align}
    C_i
    \leq \sqrt{\eta M (N-1)/N}+\sqrt{\eta (M-N)(N-1)}\leq 2\sqrt{\eta M N}.
\end{align}

\newpage
\bibliography{reference}

\begin{thebibliography}{61}
\bibitem [{eff()}]{efficient}%
  \BibitemOpen
  \href@noop {} {\bibinfo {title} {The definition of the efficient sampler is
  weaker than the one that requires for total variance distribution to be
  $\epsilon$-close in runtime \text{poly}$(n,1/\epsilon)$. {T}he depth bound
  efficiently simulated depends on the choice of the definition.}}\BibitemShut
  {Stop}%
\end{thebibliography}

\end{document}